\documentclass[12pt]{article} 

\usepackage{amssymb}
\usepackage{fullpage}
\usepackage{fancybox}
\usepackage{graphicx}

\newtheorem{theorem}{Theorem}
\newtheorem{lemma}{Lemma}
\newtheorem{definition}{Definition}
\newtheorem{corollary}{Corollary}
\newtheorem{remark}{Remark}

\newcommand{\qed}{\rule{0.5em}{1.5ex}}
\newcommand{\fqed}{{\hfill~\qed}}
\newenvironment{proof}{{\noindent \bf Proof.}}
                      {{\hfill \fqed} \vspace{1em}}

\newcommand{\IR}{\mathbb{R}}

\newcommand{\dist}{\mathord{\it dist}}
\newcommand{\ball}{\mathord{\it ball}}
\newcommand{\ann}{\mathord{\it annulus}}
\newcommand{\sannulus}{\textsf{\sc SepAnn}}
\newcommand{\ssannulus}{\textsf{\sc SparseSepAnn}}
\newcommand{\CP}{\textsf{\sc ClosestPair}}

\newcommand{\algAll}[2]{\vspace{0.5em}
  \begin{minipage}{.90\linewidth}%
    \shadowbox{%
      \begin{minipage}{\linewidth}%
        \textbf{#1}%
      \end{minipage}}
    \par%% \vspace{0.5em}
    {
      \fontencoding{OT1}\fontfamily{ppl}\small#2}
    \par\vspace{-1em}
    \noindent\rule{\linewidth}{1mm} \linebreak
  \end{minipage}
  \vspace{0.5em}
}
\title{A Simple Randomized $O(n \log n)$--Time Closest-Pair Algorithm 
       in Doubling Metrics}
\author{Anil Maheshwari\thanks{School of Computer Science,
                               Carleton University, Ottawa, Canada.
                               Research supported by NSERC.}
\and 
Wolfgang Mulzer\thanks{Institut f\"ur Informatik,
                       Freie Universit\"at Berlin, Germany.
                       Research supported in part by ERC StG 757609.}
\and 
Michiel Smid\footnotemark[1] 
}
\date{\today} 

\begin{document}

\maketitle

\begin{abstract} 
Consider a metric space $(P,\dist)$ with $N$ points whose doubling 
dimension is a 
constant. We present a simple, randomized, and recursive algorithm that 
computes, in $O(N \log N)$ expected time, the closest-pair distance in 
$P$. To generate recursive calls, we use previous results of 
Har-Peled and Mendel, and Abam and Har-Peled for computing a sparse 
annulus that separates the points in a balanced way. 
\end{abstract} 

\begin{center}
\hfill\shadowbox{
  \begin{minipage}{.80\linewidth}
    {\textsf{For a long time researchers felt that there might be a
         quadratic lower bound on the complexity of the closest-pair
         problem.}  \\
      {\hspace*{\fill}\nolinebreak[1]\hspace*{\fill}} ---
      Jon Louis Bentley, \\
      {\hspace*{\fill}\nolinebreak[1]\hspace*{\fill}} ---
      \emph{Communications of the ACM}, volume 23, page 226, 1980}
  \end{minipage}
}
\end{center}

\section{Introduction}  \label{secintro}    
The closest-pair problem is one of the oldest problems in computational 
geometry: Given a set $P$ of $N$ points in the Euclidean space $\IR^d$,
where $d \geq 1$ is a constant, compute a \emph{closest-pair} in $P$, 
i.e., a pair $p,q$ of distinct points in $P$ for which the Euclidean 
distance $\dist(p,q)$ is minimum.  

\paragraph{The algorithm of Bentley and Shamos.}
The first $O(N \log N)$--time algorithm for this problem dates back to 
1976 and is due to Bentley and Shamos~\cite{bs-dcms-76} (See also 
Bentley~\cite{b-phd-76}). When $d=2$, the algorithm is particularly 
simple and an excellent example of a ``textbook algorithm'' that 
illustrates the power of the divide-and-conquer paradigm; see 
Cormen \emph{et al.}~\cite[Section~33.4]{clrs-ia-09} and 
Kleinberg and Tardos~\cite[Section~5.4]{kt-ad-06}. 
Bentley~\cite{b-cacm-80} mentions that this algorithm, for $d=2$, is 
due to Shamos, and the idea of using divide-and-conquer was suggested by 
H.R.\@ Strong.

\begin{figure}
\centering
\includegraphics{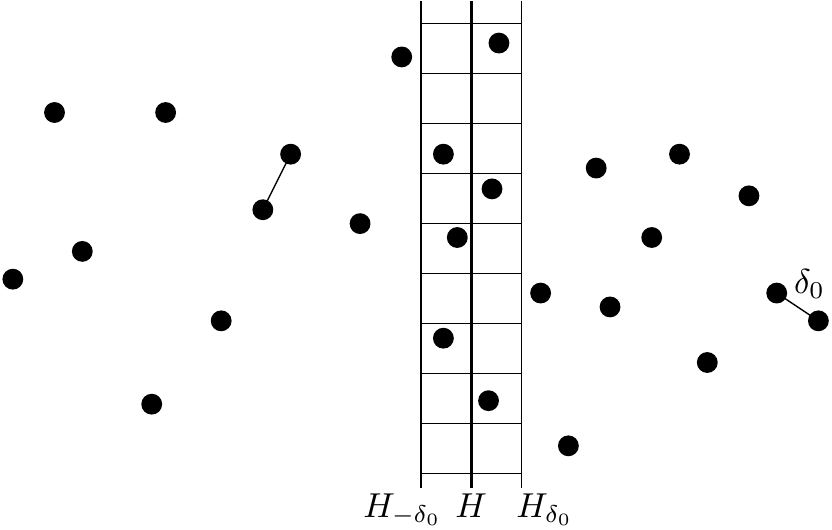}
\caption{The divide-and-conquer algorithm by Bentley and Shamos
in two dimensions: compute a hyperplane $H$ that partitions the
point set $P$ evenly; recurse on the ``left'' and on the ``right''
part; compute the closest pair in the slab between $H_{-\delta_0}$
and $H_{\delta_0}$, where $\delta_0$ is the minimum of the
two closest-pair distances to the left and to the right of $H$.}
\label{fig:bentley-shamos}
\end{figure}

We briefly describe the Bentley--Shamos algorithm. In a preprocessing 
step, for each $i=1,2,\ldots,d$, the algorithm sorts the points of $P$ 
according to their $i$-th coordinates. 

If $d=1$, the closest-pair in $P$ can easily be computed in $O(N)$ time,
by scanning the sorted sequence of elements of $P$. 

Assume that $d \geq 2$. We first introduce some notation. Let $H$ be a 
hyperplane that is orthogonal to one of the $d$ coordinate axes. For any 
real number $\delta>0$, we denote by $H_{-\delta}$ and $H_{+\delta}$ the 
two hyperplanes that are obtained by translating $H$ by a distance of 
$\delta$ to the ``left'' and ``right'', respectively.   

Bentley and Shamos prove that there exists a hyperplane $H$, such that,
for some positive constant $\alpha > 0$ that only depends on the dimension 
$d$, the following properties hold. First, at least $\alpha N$ points of 
$P$ are to the ``left'' of $H$ and at least $\alpha N$ points of $P$ are 
to the ``right'' of~$H$. Second, let $\delta_0$ be the smaller of the 
closest-pair distance to the left of $H$ and the closest-pair distance 
to the right of~$H$. Then, the \emph{slab} defined by $H_{-\delta_0}$ 
and $H_{+\delta_0}$ contains $O(N^{1-1/d})$ points of $P$. Third, 
such a hyperplane $H$ can be computed in $O(N)$ time. 
Observe that the exact value of $\delta_0$ is not known when $H$ is 
computed. However, during the computation of~$H$, we do obtain an 
upper bound on $\delta_0$. 

To compute the closest-pair distance in $P$, the algorithm recurses on 
two subproblems in~$\IR^d$, one subproblem for the points to the left of 
$H$ and one subproblem for the points to the right of $H$. Finally, the 
algorithm must consider the points inside the slab. Observe that these 
points are ``sparse'', in the sense that the number of points inside any 
hypercube with sides of length $\delta_0$ is bounded from above by a 
function that only depends on $d$, see Figure~\ref{fig:bentley-shamos}. 
Bentley and Shamos use the 
divide-and-conquer technique to solve the sparse problem, on only 
$O(N^{1-1/d})$ points, in $O(N)$ time. 

The total running time $T(N)$ of this algorithm satisfies the standard 
merge-sort recurrence 
\[ T(N) = O(N) + T(N') + T(N'') , 
\]
where $N' \leq (1-\alpha) N$, $N'' \leq (1-\alpha) N$, and $N' + N' = N$. 
It follows that the algorithm computes the closest-pair distance in $P$ 
in $O(N \log N)$ time. 

\paragraph{Our results.} 
The algorithm of Bentley and Shamos 
uses the fact that the points in the set $P$ have 
coordinates. This leads to the problem considered in this paper: Can we 
compute the closest-pair distance, by only using distances? Thus, we 
assume that $(P,\dist)$ is a metric space (to be defined in 
Section~\ref{secMS}), and we have an oracle that returns, in $O(1)$ 
time, the distance $\dist(p,q)$ for any two elements $p$ and $q$ of $P$. 

In general metric spaces, the closest-pair distance cannot be computed in 
subquadratic time: Assume that exactly one distance is equal to $1$, 
and all other distances are equal to~$2$. An easy adversary argument 
implies that any algorithm that computes the closest-pair distance 
must take $\Omega(N^2)$ time in the worst case.  

In this paper, we present a randomized algorithm that computes the 
closest-pair distance in $O(N \log N)$ expected time, for the case when 
the \emph{doubling dimension} of the metric space $(P,\dist)$ is bounded 
by a constant. Informally, this means that any ball can be covered by 
$O(1)$ balls of half the radius; the formal definition will be given in 
Section~\ref{secMS}.   

A closest-pair algorithm can be obtained from results by 
Har-Peled and Mendel~\cite{hm-fcnld-05}: They show that a well-separated 
pair decomposition of $P$ can be computed in $O(N \log N)$ expected time. 
Given this decomposition, the closest-pair distance can be obtained in 
$O(N)$ time. The drawback of this approach is that this algorithm is 
quite technical. We show that there is a very simple algorithm that 
computes the closest-pair distance in $O(N \log N)$ expected time.  
As we will see later, one of the main ingredients that we use is 
from~\cite{hm-fcnld-05}. 

Since the elements of $P$ do not have coordinates, there are no notions 
of a hyperplane or a slab. It is natural to replace these by a 
\emph{ball} and an \emph{annulus}; the latter is the subset of 
points between two concentric balls. 

\begin{figure}[t]
\centering
\includegraphics{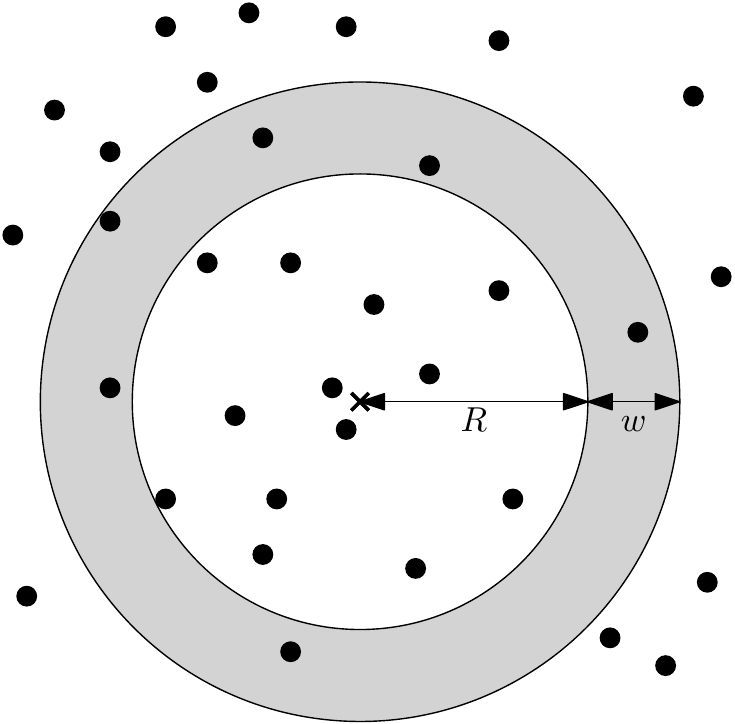}
\caption{A sparse separating annulus for a planar point
set $P$ with $N$ points: Each of the regions inside and outside the
annulus contains $\Omega(N)$ points; inside the annulus, there
are $O(\sqrt{N})$ points; and the width $w$ of the annulus is
proportional to $R / \sqrt{N}$.}
\label{fig:ssannulus}
\end{figure}

Let $d$ denote the doubling dimension of the metric space $(P,\dist)$. 
Abam and Har-Peled~\cite{ahp-ncsspd-12}, using a previous result of  
Har-Peled and Mendel~\cite{hm-fcnld-05}, show that, in $O(N)$ expected 
time, two concentric balls of radii $R$ and $R+w$ can be computed, such 
that, for some positive constant $\alpha > 0$ that only depends on $d$, 
\begin{enumerate} 
\item the ball of radius $R$ contains at least $\alpha N$ points, 
\item there are at least $\alpha N$ points outside the ball of radius 
      $R+w$, 
\item the annulus with radii $R$ and $R+w$ contains $O(N^{1-1/d})$ 
      points, and 
\item the width $w$ of this annulus is proportional to $R / N^{1/d}$. 
\end{enumerate}  
We will refer to this annulus as a \emph{sparse separating annulus},
see Figure~\ref{fig:ssannulus}. 
In Section~\ref{secSSA}, we will present a simplified version of the 
algorithm of Abam and Har-Peled~\cite{ahp-ncsspd-12} 
that computes such an annulus. 

Let $\delta$ be the closest-pair distance in $P$. A packing argument 
(see Section~\ref{secPL}) shows that the above ball of radius $R$ 
contains $O( (R/\delta)^d )$ points. Since this ball contains at least 
$\alpha N$ points, it follows that $R = \Omega(\delta \cdot N^{1/d})$. 
Thus, by choosing appropriate constants, the width $w$ of the above 
annulus is at least $\delta$. (The formal proofs will be presented in 
Section~\ref{secCP}.) Observe that, as in the Bentley--Shamos algorithm, 
the value of $\delta$ is not known when the two concentric balls are 
computed. 

Let $P_1$ be the subset of all points that are inside the ball of 
radius $R$, let $P_2$ be the subset of all points that are inside the 
annulus, and let $P_3$ be the subset of all points that are outside 
the ball of radius $R+w$. Then it suffices to recursively run the 
algorithm twice, once on $P_1 \cup P_2$, and once on $P_2 \cup P_3$. 
The expected running time of this algorithm satisfies the recurrence,   
\[ T(N) = O(N) + T(N') + T(N'') ,
\]
where $N' \leq (1-\alpha) N$, $N'' \leq (1-\alpha) N$, and 
$N' + N'' \leq N + O(N^{1-1/d})$. We will prove in Section~\ref{secsolve} 
that this recurrence solves to $T(N) = O(N \log N)$.

\section{Metric spaces and their doubling dimension}  \label{secMS}   
A \emph{metric space} is a pair $(P,\dist)$, where $P$ is a non-empty 
set and $\dist: P \times P \rightarrow \IR$ is a function such that for 
all $x$, $y$, and $z$ in $P$, 
\begin{enumerate} 
\item $\dist(x,x) = 0$,
\item $\dist(x,y) > 0$ if $x \neq y$, 
\item $\dist(x,y) = \dist(y,x)$, and 
\item $\dist(x,z) \leq \dist(x,y) + \dist(y,z)$. 
\end{enumerate} 
The fourth property is called the \emph{triangle inequality}. 
We refer to $\dist(x,y)$ as the \emph{distance} between $x$ and $y$. 
We only consider metric spaces in which the set $P$ is finite. We call 
the elements of $P$ \emph{points}. 

If $p \in P$ is a point, $S \subseteq P$ is a subset of $P$, 
and $R$, $R'$ are 
real numbers with $R' \geq R \geq 0$, then the \emph{ball} in $S$ with 
\emph{center} $p$ and \emph{radius} $R$ is the set 
\[ \ball_S(p,R) = \{ x \in S : \dist(p,x) \leq R \} ,
\]
and the \emph{annulus} in $S$ with \emph{center} $p$, 
\emph{inner radius} $R$, and \emph{outer radius} $R'$ is the set 
\[ \ann_S(p,R,R') = \{ x \in S : R < \dist(p,x) \leq R' \} .
\]
The \emph{closest-pair} distance in $S$ is 
\[ \delta(S) = 
    \left\{ 
      \begin{array}{ll} 
         \infty & \mbox{if $|S| \leq 1$,} \\ 
         \min \{ \dist(x,y) : x \in S, y \in S , x \neq y \} 
               & \mbox{if $|S| \geq 2$.}  
      \end{array} 
    \right. 
\] 

The doubling dimension of a metric space was introduced by 
Assouad~\cite{a-pldr-83}; see also Heinonen~\cite{h-lams-01}:

\begin{definition} 
\emph{
Let $(P,\dist)$ be a finite metric space and let $\lambda$ be the 
smallest integer such that the following is true: For every point $p$ in 
$P$ and every real number $R>0$, $\ball_P(p,R)$ can be covered by at most 
$\lambda$ balls in $P$ of radius $R/2$. The \emph{doubling dimension} 
of $(P,\dist)$ is defined to be $\log \lambda$.
}
\end{definition}

The doubling dimension is in the interval $[1, \log|P|]$ and, in 
general, is not an integer. For example, if $\dist$ is the Euclidean 
distance function in $\IR^2$, the doubling dimension is $\log 7$,
whereas in $\IR^d$, the doubling dimension is $\Theta(d)$. 
The discrete metric space $(P,\dist)$ in which the distance between any two 
distinct points is equal to $1$ has doubling dimension $\log |P|$.

\subsection{The doubling dimension of a subset} 
Our algorithm for computing the closest-pair distance in $P$ uses 
recursion. In a recursive call, the algorithm is run on a subset $S$ 
of $P$. We show below that the doubling dimension of $S$ may not be 
the same as that of $P$. 

Let $(P,\dist)$ be a metric space, let $d$ be its doubling dimension,
and let $S$ be a non-empty subset of $P$. To determine the doubling 
dimension of 
$(S,\dist)$,\footnote{With a slight abuse of notation, when writing 
$(S,\dist)$, we consider $\dist$ to be the restriction of the distance 
function to the set $S \times S$.}
we have to cover any ball $\ball_S(p,R)$, with $p \in S$ and $R>0$, by 
balls in $S$ of radius $R/2$ that are centered at points of~$S$. The 
number of such balls may be larger than $2^d$. 

\begin{figure}
  \centering
  \includegraphics{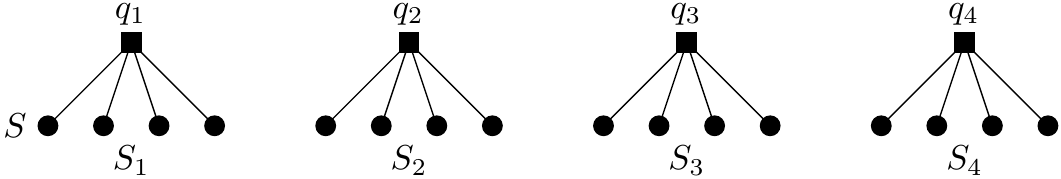}
  \caption{A metric space $P$ of 20 points and a subset
  $S = \bigcup_{i = 1}^4 S_i$ of $P$ with strictly
  smaller doubling dimension. For $i = 1, \dots, 4$, the distance
  between $q_i$ and all points in $S_i$ is $1$. All other 
  distances between pairs of distinct points are $2$.}
  \label{fig:ddimsubset}
\end{figure}

To give an example, let $n$ be a positive integer and let $(S,\dist)$ 
be the metric space of size $n^2$ with $\dist(x, y)=2$ for all distinct 
points $x$ and $y$ in $S$. The doubling dimension $d_S$ of $(S,\dist)$ 
is equal to 
\[ d_S = \log |S| = 2 \log n .
\] 
Partition $S$ into subsets $S_1,S_2,\ldots,S_n$, each consisting of $n$ 
points. Let $q_1,q_2,\ldots,q_n$ be new points, and let 
\[ P = S \cup \{ q_1,q_2,\ldots,q_n \} .
\] 
(For an illustration with $n=4$, refer to Figure~\ref{fig:ddimsubset}.) 
For any two points $x$ and $y$ in $P$, define 
\[ \dist(x,y) = 
    \left\{ 
     \begin{array}{ll} 
       0 & \mbox{if $x=y$,} \\ 
       1 & \mbox{if there is an $i$ such that $x=q_i$ and $y \in S_i$, 
                 or $x \in S_i$ and $y=q_i$,} \\
       2 & \mbox{otherwise.} \\
     \end{array} 
    \right. 
\] 
Since all distances between distinct points are $1$ or $2$, it follows
that $(P,\dist)$ fulfills the triangle inequality. Hence,
$(P, \dist)$ is a metric space. We will prove 
below that the doubling dimension $d$ of this metric space is equal to 
\[ d = \log (n+1) . 
\]
Thus, for large values of $n$, the ratio $d_S/d$ converges to $2$. 

To determine the doubling dimension of $(P,\dist)$, let $p$ be a point 
of $P$, let $R > 0$ be a real number, and let $B = \ball_P(p,R)$. 
If $R \in (0, 1)$, then $B$ is a singleton set, which is covered by the ball 
$\ball_P(p,R/2)$. If $R \in [2, \infty)$, then $B = P$, which 
is covered by the 
$n$ balls in $P$ of radius $R/2$ that are centered at 
$q_1,q_2,\ldots,q_n$. If $R \in [1, 2)$, then 
$B = \{ q_i \} \cup S_i$ for some $i$. In this case, $B$ can only be 
covered by the $n+1$ balls in $P$ of radius $R/2$ that are centered at the 
points of $B$. Thus, for each case, we have shown that $B$ can be covered 
by at most $n+1$ balls in $P$ of radius $R/2$, and for some $B$, we
need $n+1$ 
such balls. This proves that $d = \log(n+1)$. 

The following lemma states that the doubling dimension of a subset $S$ 
of $P$ is always at most twice the doubling dimension of $P$.  

\begin{figure}
\centering
\includegraphics{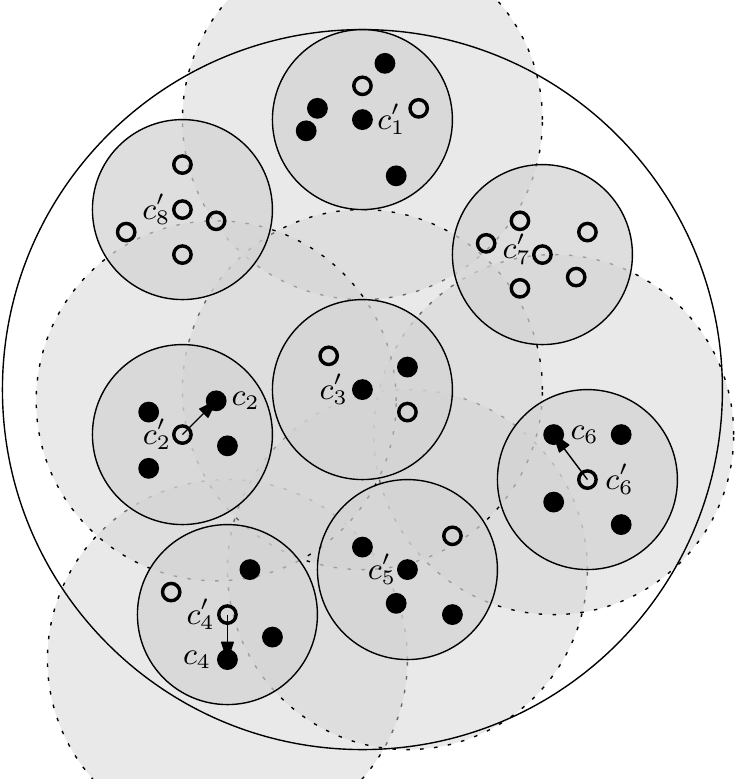}
\caption{Illustration of the proof of Lemma~\ref{lem:ddimsubset}.
The points of $S$ are solid; the points of $P \setminus S$
are empty.  The ball $B'$ can be covered in $P$ by $8$ balls $B'_1,
\ldots, B_8'$ with centers $c_1', \ldots, c_8'$. For $B'_1, \ldots, B'_6$, 
the intersection with $S$ is nonempty. The centers $c_1'$, $c_3'$, and $c_5'$ 
are also in $S$, the centers $c_2'$, $c_5'$, and $c_6'$ must be moved.
This increases the covering radius to 
$R/2$.}
\label{fig:ddimsubset2}
\end{figure}

\begin{lemma} 
\label{lem:ddimsubset}
Let $(P,\dist)$ be a metric space, let $d$ be its doubling dimension,
and let $S$ be a non-empty subset of $P$. Then the metric space 
$(S,\dist)$ has doubling dimension at most $2d$. 
\end{lemma} 
\begin{proof} 
Let $p$ be a point in $S$, let $R>0$ be a real number, and consider the 
ball $B = \ball_S(p,R)$ in $S$. Let $B' = \ball_P(p,R)$ be the 
corresponding ball in $P$. By applying the definition of doubling 
dimension twice, we can cover $B'$ by balls $B'_i$, 
for $1 \leq i \leq 2^{2d}$, in $P$, each having radius $R/4$.  
Let $k$ be the number of indices $i$ for which 
$B'_i \cap S \neq \emptyset$. We may assume, without loss of generality, 
that $B'_i \cap S \neq \emptyset$ for all $i$ with $1 \leq i \leq k$, 
and $B'_i \cap S = \emptyset$ for all $i$ with $k+1 \leq i \leq 2^{2d}$. 
For $i=1,2,\ldots,k$, let $c'_i \in P$ be the center of $B'_i$, let 
\[ c_i = 
    \left\{ 
      \begin{array}{ll} 
        c'_i & \mbox{if $c'_i \in S$,} \\ 
        \mbox{an arbitrary point in $B'_i \cap S$} & 
               \mbox{if $c'_i \not\in S$,} 
      \end{array}
    \right. 
\] 
and let $B_i = \ball_S(c_i,R/2)$, see Figure~\ref{fig:ddimsubset2}.

We claim that the balls $B_i$ in $S$, $1 \leq i \leq k$, cover the ball 
$B$. To prove this, let $q$ be a point in $B$. Then, $q \in B'$ and, 
thus, there is an index $i$, $1 \leq i \leq k$, with 
$q \in B'_i$. Since 
\[ \dist(c_i,q) \leq \dist(c_i,c'_i) + \dist(c'_i,q) \leq R/4 + R/4  
     = R/2 ,
\]
the point $q$ is in the ball $B_i$.
We have shown that any ball in $S$ of radius $R$ can be covered by at 
most $2^{2d}$ balls in $S$ of radius $R/2$. 
\end{proof} 

\subsection{The packing lemma}  \label{secPL} 
Consider a metric space $(P,\dist)$ whose doubling dimension is 
``small'', and a ball $B$ in $P$ whose radius $R$ is proportional to 
the closest-pair distance $\delta(P)$. By repeatedly applying the 
definition of doubling dimension, we can cover $B$ by a ``small'' number 
of balls of radius less than $\delta(P)$. Since each of these smaller
balls contains only one point, the original ball $B$ cannot contain ``many'' 
points. The following lemma formalizes this. 

\begin{lemma}  \label{lemball}  
Let $(P,\dist)$ be a finite metric space with $|P| \geq 2$ and 
doubling dimension $d$.
Let $\delta$ be the closest-pair distance 
in $P$. Then, for any point $p$ in $P$ and any real number 
$R \geq \delta/2$, 
\[ 
  |\ball_P(p,R)| \leq (4R/\delta)^d . 
\] 
\end{lemma} 

\begin{proof}
Set $k = \lceil \log (2R/\delta) \rceil$. Then, $k \geq 0$ and 
$2R/\delta \leq 2^k < 4R/\delta$. We apply the definition of 
doubling dimension $k$ times in order to cover $\ball_P(p,R)$ by 
$2^{kd} \leq (4R/\delta)^d$ balls of radius $R/2^k < \delta$. 
Each of these $2^{kd}$ balls contains exactly one point of $P$, namely 
its center. 
\end{proof} 

\section{Computing a sparse separating annulus} 
\label{secSSA} 

Throughout this section, $(P,\dist)$ is a finite metric space, $d$ 
denotes its doubling dimension, $S$ is a non-empty subset of $P$, and 
$n$ denotes the size of $S$. Observe that $d$ will always refer to 
the doubling dimension of the entire metric space $(P,\dist)$. 

In this section, we present a simplified variant of the algorithm of 
Abam and Har-Peled~\cite{ahp-ncsspd-12} to compute the sparse 
separating annulus that was mentioned in Section~\ref{secintro}. 

\subsection{Computing a separating annulus} 

Let $\mu \geq 1$ be a real number (possibly depending on $n$) and set 
$c= 2 (8\mu)^d$. Assume that $n \geq c+1$. As a first step, we give a 
randomized algorithm that computes a point $p$ in $S$ and a real number 
$R'>0$, such that $|\ball_S(p,R')| \geq n/c$ and 
$|\ball_S(p,\mu R')| \leq n/2$. This algorithm is due to 
Har-Peled and Mendel~\cite[Lemma~2.4]{hm-fcnld-05}; see also 
Abam and Har-Peled~\cite[Lemma~2.6]{ahp-ncsspd-12}. 
In order to be self-contained, we present the algorithm and its analysis. 

The algorithm chooses a uniformly random point $p$ in $S$ and 
computes the smallest radius $R_p$ such that $\ball_S(p,R_p)$ contains 
at least $n/c$ points. Then it checks if $\ball_S(p,\mu R_p)$ contains 
at most $n/2$ points. If this is the case, the algorithm returns $p$ 
and $R_p$. Otherwise, the algorithm is repeated. The pseudocode for this 
algorithm is given below. 
  
\algAll{Algorithm $\sannulus(S,n,d,\mu,c)$}{
{\bf Comment:} The input is a subset $S$, of size $n$, of a metric space 
of doubling dimension $d$, and real numbers $\mu \geq 1$ and $c>1$. 
If $c = 2 (8\mu)^d$ and $n \geq c+1$, then the algorithm returns a point 
$p$ in $S$ and a real number $R'>0$ that satisfy the two properties in 
Lemma~\ref{lemann}. 

\begin{quote}
\begin{tabbing}
{\bf repeat} \= $p =$ uniformly random point in $S$; \\
\> $R_p = \min \{ r>0 : | \ball_S(p,r) | \geq n/c \}$ \\ 
{\bf until} $| \ball_S(p, \mu R_p) | \leq n/2$; \\
$R' = R_p$; \\ 
return $p$ and $R'$  
\end{tabbing}
\end{quote}
}

\begin{lemma}  \label{lemann}  
Let $\mu \geq 1$ be a real number (possibly depending on $n$) and set 
$c = 2 (8\mu)^d$. Assume that $n \geq c+1$. Algorithm 
$\textup{\sannulus}(S,n,d,\mu,c)$ 
has expected running time $O(cn)$. It returns a 
point $p$ in $S$ and a real number $R'>0$, such that 
\begin{enumerate} 
\item $| \ball_S(p,R') | \geq n/c$ and 
\item $| \ball_S(p,\mu R') | \leq n/2$. 
\end{enumerate} 
\end{lemma} 
\begin{proof}
Let $p$ be a point in $S$. Since $n \geq c+1$, we have 
$|\ball_S(p,R_p)| \geq n/c > 1$. Therefore, $\ball_S(p,R_p)$ contains 
at least two points of $S$, which means that $R_p > 0$. The radius $R_p$ 
can be found in $O(n)$ time, by selecting the $\lceil n/c \rceil$-th 
smallest element in the sequence of distances between $p$ and all 
points of $S$ (including $p$ itself). By scanning this sequence, we can 
compute $|\ball_S(p,\mu R_p)|$ in $O(n)$ time. Thus, one iteration of 
the algorithm takes $O(n)$ time. 

We say that a point $p$ in $S$ is \emph{good}, if 
$|\ball_S(p,\mu R_p)| \leq n/2$. We will prove below that a uniformly 
random point of $S$ is good with probability at least $1/c$. This will 
imply that the expected number of iterations of the algorithm is at most 
$c$ and, therefore, the expected running time is $O(cn)$. 

Consider a ball in $P$ of minimum radius that contains at least $n/c$ 
points of $S$ and that is centered at a point of $P$. 
Let $q \in P$ be the center of this ball and let $R$ be its radius. 
We claim that every point in $\ball_S(q,R)$ is good. This will imply 
that a uniformly random point in $S$ has probability at least $1/c$ of 
being good. See Figure~\ref{fig:ann} for an illustration of the argument.

\begin{figure}
\centering
\includegraphics{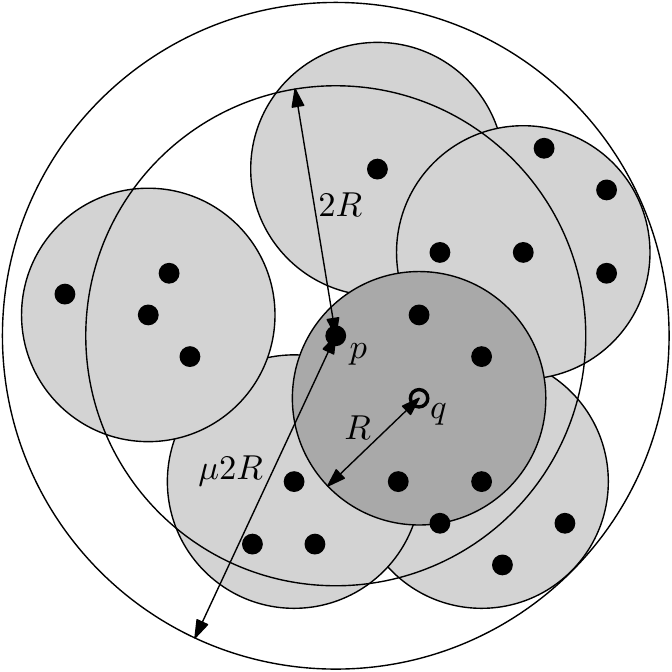}
\caption{Illustration of the proof of Lemma~\ref{lemann}. The point
$q \in P$ and the radius $R$ are such that $\ball_P(q, R)$ is the
minimum-radius ball in $P$ that contains at least $n/c$ points from $S$.
If we pick an arbitrary point $p \in \ball_S(q, R)$, then
$\ball_S(p, 2R)$ covers $\ball_S(q, R)$ and hence contains at least
$n/c$ points. The ball $\ball_P(p, \mu 2R)$ can be covered by $c/2$
balls in $P$ of radius $R$, and hence it contains at most $n/2$ points from
$S$.}
\label{fig:ann}
\end{figure}

To prove the claim, let $p$ be a point in $\ball_S(q,R)$. We will show 
that $|\ball_S(p,\mu R_p)| \leq n/2$. We first observe that 
\[ \ball_S(q,R) \subseteq \ball_S(p,2R) .
\] 
Indeed, if $x \in \ball_S(q,R)$, then 
\[ \dist(p,x) \leq \dist(p,q) + \dist(q,x) \leq R + R = 2 R  
\] 
and, therefore, $x \in \ball_S(p,2R)$. 
It follows that 
\[ |\ball_S(p,2 R)| \geq |\ball_S(q,R)| \geq n/c ,
\]
which implies that 
\[ R_p \leq 2 R . 
\] 

Let $k = \lceil \log (4\mu) \rceil$. By the definition of doubling 
dimension, we can cover $\ball_P(p, \mu 2R)$ by 
$2^{kd} < 2^{(\log(4\mu) + 1)d} = (8 \mu)^d = c/2$ balls in $P$ of 
radius $ \mu 2R/2^k < R$. By the definition of~$R$, each of these 
(at most) $c/2$ balls contains less than $n/c$ points of $S$. Therefore,
\[ 
  |\ball_S(p,\mu R_p)| \leq |\ball_S(p, \mu 2R)| < 
     c/2 \cdot n/c = n/2 . 
\]
Thus, we have shown that every point in $\ball_S(q,R)$ is good. 
\end{proof} 

\begin{remark} \label{rem1} 
\emph{
Consider the parameters $\mu$, $c$, and $k$ in Lemma~\ref{lemann} and its 
proof. If $\log (4\mu)$ is not an integer, then we can take 
$k = \lceil \log (2 \mu) \rceil$ and reduce the value of $c$ to 
$2 (4\mu)^d$.  
}
\end{remark}

\subsection{A refinement of the algorithm} 
Algorithm $\sannulus(S,n,d,\mu,c)$ returns a point $p$ and a real 
number $R'>0$, such that $| \ball_S(p,R') | \geq n/c$ and 
$| \ball_S(p,\mu R') | \leq n/2$. The annulus $\ann_S(p,R',\mu R')$ may 
contain $\Theta(n)$ points. In this section, we present a refinement of 
this algorithm that outputs an annulus that contains a ``small'' number 
of points of $S$. Our algorithm is a simplified version of an algorithm 
due to Abam and Har-Peled~\cite[Lemma~2.7]{ahp-ncsspd-12}. 

The refined algorithm takes as input an integer $t \geq 1$ that may 
depend on $n$. First it runs algorithm $\sannulus(S,n,d,\mu,c)$ with 
$\mu = e$ and (by Remark~\ref{rem1}) $c = 2 (4e)^d$. 
Consider the 
output $p$ and $R'$. Recall that (since $n \geq c + 1$) we have $R' > 0$. Let 
\[ R_i = (1+1/t)^i \cdot R' 
\]
for $i = 0,1,\ldots,t$, and 
\[ A_i = \ann_S(p,R_{i-1},R_i) 
\] 
for $i = 1,2,\ldots,t$. The inequality $1+x \leq e^x$, which is valid 
for all real numbers $x$, implies that, for each $i$ with 
$0 \leq i \leq t$, 
\[ R_i \leq \left( e^{1/t} \right)^i \cdot R' = 
         e^{i/t} \cdot R' \leq e R' . 
\] 

\begin{figure}
\centering
\includegraphics{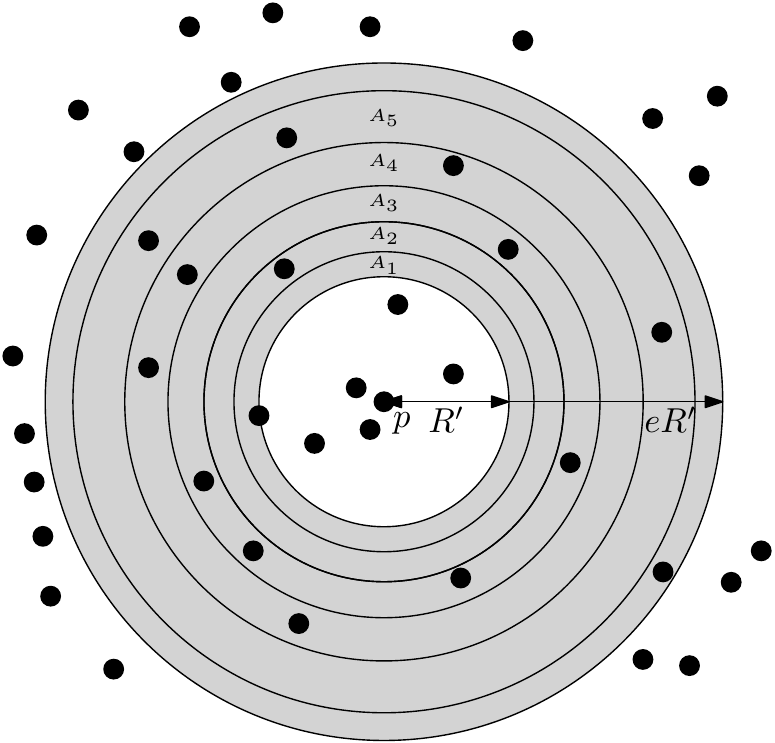}
\caption{The annulus $\ann_S(p, R', eR')$ contains $t = 5$ annuli
$A_1, A_2, \ldots, A_5$. At least one of them contains at most 
$n/10$ points, and at least $3$ of them contain at most $n/5$
points.}
\label{fig:ssannulus2}
\end{figure}

\noindent 
Thus, the $t$ annuli $A_i$ are contained in $\ann_S(p,R',eR')$, 
see Figure~\ref{fig:ssannulus2}.
Observe 
that they are pairwise disjoint and, together, contain at most $n/2$ 
points of $S$. Therefore, there is an $i$ such that $|A_i| \leq n/(2t)$. 
We can compute $|A_1|, |A_2|, \ldots, |A_t|$ and, thus, the smallest 
of these values, as follows: Any point $x$ in $S$ with 
$R' = R_0 < \dist(p,x) \leq R_t$ is contained in $A_j$, where 
\[ j = \left\lceil \frac{\log(\dist(p,x)/R')}{\log(1+1/t)} \right\rceil . 
\] 
Thus, by scanning the sequence of distances between $p$ and all points 
of $S$, we can compute, in $O(n)$ time, an index $i$ such that 
$|A_i| \leq n/(2t)$. This is the approach of 
Abam and Har-Peled~\cite{ahp-ncsspd-12}. 

Our simplification uses the fact that, on average, one annulus $A_i$  
contains at most $n/(2t)$ points of $S$ and, thus, by Markov's
inequality, at least $t/2$ of 
these annuli contain at most $n/t$ points of $S$. The algorithm finds 
such an annulus $A_i$ by repeatedly choosing a uniformly random 
element $i$ from $\{1,2,\ldots,t\}$. As soon as $|A_i| \leq n/t$, the 
algorithm returns $p$ and $R_{i-1}$. The pseudocode for this algorithm 
is given below.  

\algAll{Algorithm $\ssannulus(S,n,d,t)$}{
{\bf Comment:} The input is a subset $S$, of size $n \geq 2 (4e)^d +1$, 
of a metric space of doubling dimension $d$, and an integer $t \geq 1$. 
The algorithm returns a point $p$ in $S$ and a real number $R>0$ that 
satisfy the three properties in Lemma~\ref{lemrann}. 

\begin{quote}
\begin{tabbing}
$c = 2 (4e)^d$; \\ 
let $p \in S$ and $R'>0$ be the output of algorithm 
$\sannulus(S,n,d,e,c)$; \\
{\bf repeat} \= $i=$ uniformly random element in $\{1,2,\ldots,t\}$; \\ 
         \> $s = |A_i|$ \\
{\bf until} $s \leq n/t$; \\ 
$R = R_{i-1}$; \\
return $p$ and $R$ 
\end{tabbing}
\end{quote}
}

\begin{lemma}  \label{lemrann}  
Let $t \geq 1$ be an integer (possibly depending on $n$) and let 
$c = 2 (4e)^d$. Assume that $n \geq c+1$. Algorithm 
$\textup{\ssannulus}(S,n,d,t)$ has expected running time $O(cn)$. It returns a 
point $p$ in $S$ and a real number $R>0$, such that 
\begin{enumerate} 
\item $| \ball_S(p,R) | \geq n/c$, 
\item $| \ann_S(p,R,(1+1/t)R) | \leq n/t$, and  
\item $| S \setminus \ball_S(p,(1+1/t)R) | \geq n/2$. 
\end{enumerate} 
\end{lemma} 
\begin{proof}
Consider the output $p$ and $R'$ of algorithm $\sannulus(S,n,d,e,c)$. 
We have seen above that the annuli $A_1,A_2,\ldots,A_t$ are contained 
in $\ball_S(p,eR')$ and, thus, together, contain at most $n/2$ points 
of $S$. Moreover, at least $t/2$ of these annuli contain at most $n/t$ 
points of $S$. Therefore, in one iteration of the repeat-until-loop in algorithm 
$\ssannulus(S,n,d,t)$, the size of $A_i$ is at most $n/t$ with 
probability at least $1/2$. It follows that the expected number of
iterations of this repeat-until-loop is at most two. Since one iteration takes 
$O(n)$ time (by scanning the sequence of distances between $p$ and all
points of $S$), the entire repeat-until-loop takes expected time $O(n)$. This, 
together with Lemma~\ref{lemann}, implies that the expected running time 
of algorithm $\ssannulus(S,n,d,t)$ is $O(cn)$.  

Consider the output $p$ and $R = R_{i-1}$. We have 
\[ |\ball_S(p,R)| \geq |\ball_S(p,R')| \geq n/c 
\]
and 
\[ |\ann_S(p,R,(1+1/t)R)| = |A_i| \leq n/t, 
\]
proving the first two properties in the lemma. Since 
\[ |\ball_S(p,(1+1/t)R)| = |\ball(p,R_i)| \leq |\ball(p,e R')| \leq n/2 ,
\]
we have 
\[ |S \setminus \ball_S(p,(1+1/t)R)| \geq n/2 , 
\]
proving the third property in the lemma. 
\end{proof}

\section{The closest-pair algorithm} 
\label{secCP} 

Let $(P,\dist)$ be a finite metric space, let $N=|P|$, let $d$ be its 
doubling dimension, and let $\delta$ be its closest-pair distance. The 
recursive algorithm $\CP(P,N,d)$ returns the value of $\delta$. In a 
generic call, the algorithm takes a subset $S$ of $P$ as input and 
returns a value $\delta_0$ that is at least $\delta$. If the 
closest-pair distance in $S$ is equal to $\delta$, then 
$\delta_0 = \delta$. As before, in each recursive call, $d$ refers to the 
doubling dimension of the entire metric space $(P,\dist)$. 

\subsection{The algorithm}

\begin{figure}
\centering
\includegraphics{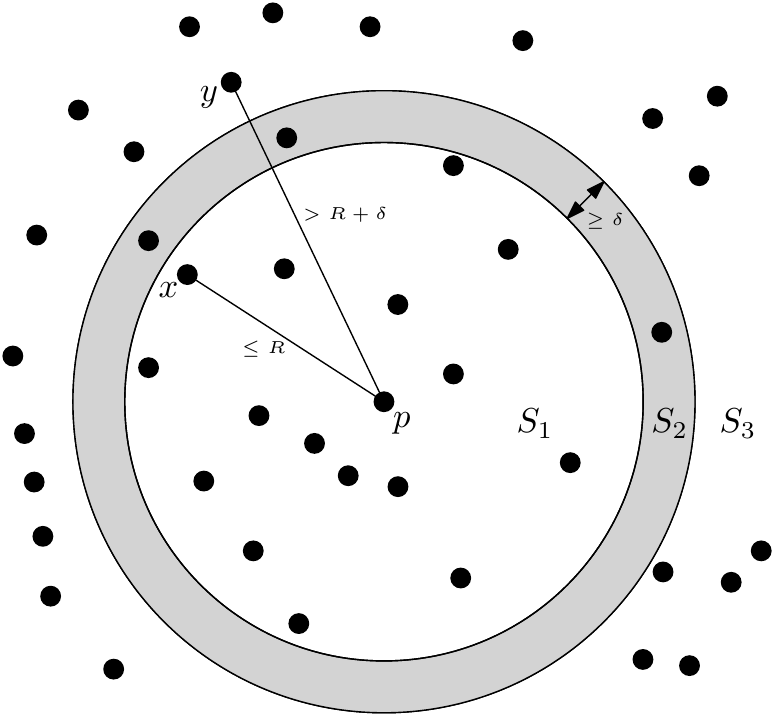}
\caption{The recursion of the closest pair algorithm.
The annulus $\ann_S(p, R, (1 + 1/t)R)$ splits $S$ into three
point sets $S_1$, $S_2$, and $S_3$. It has width at least
$\delta$ and contains $O(n^{1-1/d})$ points. The distance between
any point $x$ in $S_1$ and any point $y$ in $S_3$ is more than $\delta$.}
\label{fig:cp}
\end{figure}

Let $S$ be a subset of $P$ and let $n=|S|$. If $n$ is small, then 
algorithm $\CP(S,n,d)$ computes the closest-pair distance in $S$ by 
brute force. Otherwise, the algorithm runs $\ssannulus(S,n,d,t)$,
where $t$ is proportional to $n^{1/d}$. Consider the output $p \in S$ 
and $R>0$. By Lemmas~\ref{lemball} and~\ref{lemrann}, 
$\ann_S(p,R,(1+1/t)R)$ contains at most $n/t = O(n^{1-1/d})$ points 
of $S$ and its width is at least (the unknown value of) $\delta$,
see Figure~\ref{fig:cp}. 
Therefore, it suffices to generate two recursive calls, one on the 
points in $\ball_S(p,(1+1/t)R)$ and one on the points outside 
$\ball_S(p,R)$. The pseudocode is given below.    

\algAll{Algorithm $\CP(S,n,d)$}{
{\bf Comment:} The input is a subset $S$, of size $n \geq 2$, of the 
metric space $(P,\dist)$ of doubling dimension $d$. The algorithm 
returns a real number $\delta_0$ that satisfies the two properties in 
Lemma~\ref{lemcorrect}. 

\begin{quote}
\begin{tabbing}
{\bf if} $n < 2 (16e)^d$ \\ 
{\bf then} compute the closest-pair distance $\delta_0$ in $S$ by 
           brute force \\
{\bf else} \= $t = \lfloor \frac{1}{16e} (n/2)^{1/d} \rfloor$; \\ 
           \> let $p \in S$ and $R>0$ be the output of algorithm 
              $\ssannulus(S,n,d,t)$; \\ 
           \> $S_1 = \ball_S(p,R)$; \\
           \> $S_2 = \ann_S(p,R,(1+1/t)R)$; \\
           \> $S_3 = S \setminus (S_1 \cup S_2)$; \\
           \> $n' = | S_1 \cup S_2 |$; \\
           \> $n'' = | S_2 \cup S_3 |$; \\
           \> $\delta' = \CP(S_1 \cup S_2,n',d)$; \\
           \> $\delta'' = \CP(S_2 \cup S_3,n'',d)$; \\
           \> $\delta_0 = \min (\delta',\delta'')$ \\ 
{\bf endif}; \\
return $\delta_0$ 
\end{tabbing}
\end{quote}
}

Recall that $c = 2(4e)^d$ in algorithm $\ssannulus(S,n,d,t)$. 
Therefore, 
\begin{equation} \label{eqtc}  
  t = \left\lfloor \frac{1}{4} (n/c)^{1/d} \right\rfloor . 
\end{equation} 

Before we prove the correctness of algorithm $\CP$, we show that it 
terminates. Assume that $n \geq 2(16e)^d$. Then, $n \geq c+1$ and, by 
Lemma~\ref{lemrann}, $|S_1| \geq 2$ and $|S_3| \geq 2$. It follows 
that both $n'$ and $n''$ are at most $n-2$ and, thus, both recursive 
calls are on sets of sizes less than $n$. 

\begin{lemma}  \label{lemcorrect} 
Let $\delta$ be the closest-pair distance in $P$, let $S$ be a subset 
of $P$, let $n \geq 2$ be the size of $S$, and let $\delta_0$ be the 
output of algorithm $\textup{\CP}(S,n,d)$. Then, 
\begin{enumerate} 
\item $\delta_0 \geq \delta$ and 
\item if $\delta(S) = \delta$, then $\delta_0 = \delta$. 
\end{enumerate} 
\end{lemma} 
\begin{proof} 
The first claim holds, because the output $\delta_0$ is always the 
distance between some pair of distinct points in $S$. We prove the 
second claim by induction on $n$. This second claim obviously holds if 
$2 \leq n < 2 (16e)^d$. Assume that 
$n \geq 2 (16e)^d = 4^d c$ and $\delta(S) = \delta$. Moreover, 
assume that the second claim holds for all subsets of $S$ containing at
least two and less than $n$ points. Observe that $t \geq 1$. 

Consider the output $p \in S$ and $R>0$ of algorithm 
$\ssannulus(S,n,d,t)$. By Lemma~\ref{lemrann}, 
$|\ball_S(p,R)| \geq n/c >1$, which implies that $\ball_S(p,R)$ contains 
at least two points (with $p$ being one of them). It follows that 
$R \geq \delta$. Thus, by Lemma~\ref{lemball}, 
\[ |\ball_S(p,R)| \leq |\ball_P(p,R)| \leq (4R/\delta)^d . 
\]
By combining the two inequalities on $|\ball_S(p,R)|$, we get 
\[ n/c \leq (4R/\delta)^d ,
\]
which, using (\ref{eqtc}), implies that 
\[ R \geq (\delta/4) \cdot (n/c)^{1/d} \geq \delta t . 
\] 
The width of $\ann_S(p,R,(1+1/t)R)$ is equal to $R/t$, which is at least 
$\delta = \delta(S)$. It follows that the closest-pair distance in $S$ 
cannot be between one point in $S_1$ and one point in $S_3$. To prove 
this, let $x$ be a point in $S_1$ and let $y$ be a point in $S_3$. Then 
$\dist(p,x) \leq R$ and $\dist(p,y) > (1+1/t)R$, see Figure~\ref{fig:cp}. Thus, 
\[ (1+1/t)R < \dist(p,y) \leq \dist(p,x) + \dist(x,y) 
      \leq R + \dist(x,y) ,
\]
which implies $\dist(x,y) > R/t \geq \delta = \delta(S)$. 
It follows that the closest-pair distance in $S$ is within the set 
$S_1 \cup S_2$ or within the set $S_2 \cup S_3$. By the first claim 
in the lemma, both $\delta'$ and $\delta''$ are at least $\delta$. By 
the induction hypothesis, at least one of $\delta'$ and $\delta''$ is 
equal to $\delta$.  
\end{proof} 

Lemma~\ref{lemcorrect}, with $S=P$, proves that algorithm $\CP(P, N, d)$ 
returns the closest-pair distance in the set $P$: 

\begin{corollary} 
Let $(P,\dist)$ be a metric space of size $N \geq 2$, and let $d$ be 
its doubling dimension. The output of algorithm $\textup{\CP}(P,N,d)$ is 
the closest-pair distance in~$P$.  
\end{corollary} 

It remains to analyze the expected running time of the algorithm. 
For any integer $n \geq 2$, let $T(n)$ denote the maximum expected 
running time of algorithm $\CP(S,n,d)$, on any subset $S$ of $P$ of size 
$n$. Below, we derive a recurrence for $T(n)$. 

Assume that $n \geq 2(16e)^d = 4^d c$. Consider the sets $S_1$, $S_2$, 
and $S_3$ that are computed in the call to $\CP(S,n,d)$. By 
Lemma~\ref{lemrann}, $|S_1| \geq n/c$, $|S_2| \leq n/t$, and 
$|S_3| \geq n/2$. Thus, the values of $n' = | S_1 \cup S_2 |$ and 
$n'' = | S_2 \cup S_3 |$ satisfy 
\begin{equation}  \label{eq1}  
     2 \leq n' \leq (1-1/c) n , 
\end{equation} 
\begin{equation}  \label{eq2}  
     2 \leq n'' \leq (1-1/c) n , 
\end{equation} 
and 
\begin{equation}  \label{eq3}  
     n' + n'' \leq n + n/t . 
\end{equation} 
Observe that even though $n'$ and $n''$ are random variables, their 
values always satisfy (\ref{eq1})--(\ref{eq3}). 

By Lemma~\ref{lemrann}, the expected running time of algorithm 
$\CP(S,n,d)$ is equal to the sum of $O(cn)$ and the total expected times 
for the two recursive calls. We assume for simplicity that the constant 
in $O(cn)$ is equal to $1$. Thus, we have 
\begin{equation}  \label{eqrec} 
 T(n) \leq cn + \max_{n',n''} \left( T(n') + T(n'') \right) ,
\end{equation} 
where the maximum ranges over all $n'$ and $n''$ that satisfy 
(\ref{eq1})--(\ref{eq3}). 

If we replace (\ref{eq3}) by $n' + n'' \leq n$, then (\ref{eqrec}) 
is the standard merge-sort recurrence, whose solution is $O(n \log n)$. 
In Section~\ref{secsolve}, we will prove that, even with (\ref{eq3}),
$T(n) = O(n \log n)$, where the constant factor depends only on the 
doubling dimension of $P$. This will prove the main result of this paper: 

\begin{theorem} 
Let $(P,\dist)$ be a metric space of size $N \geq 2$, and let $d$ be 
its doubling dimension. Assume that $d$ does not depend on $N$. 
The closest-pair distance in $P$ can be computed in 
$O(N \log N)$ expected time. The constant factor in this time bound   
depends only on $d$. 
\end{theorem} 

\subsection{Solving the recurrence}  \label{secsolve}  
Throughout this section, we assume for simplicity that $d$ is an integer. 
(If this is not the case, then we replace $d$ by $\lceil d \rceil$.)  
Before we turn to the recurrence (\ref{eqrec}), we derive some 
inequalities that will be used later. 

Recall the definition of $t$, see (\ref{eqtc}). If 
$n \geq 2(32e)^d = 8^d c$, then 
\[ t = \left\lfloor \frac{1}{4} (n/c)^{1/d} \right\rfloor \geq  
    \frac{1}{4} (n/c)^{1/d} - 1 \geq 
    \frac{1}{8} (n/c)^{1/d} ,  
\]
which implies that 
\begin{equation} \label{eq4} 
     n/t \leq 8 c^{1/d} n^{1-1/d} . 
\end{equation} 
Since 
\[ \lim_{n \rightarrow \infty} \frac{n}{\ln^d n} = \infty ,
\]
there exists an $N_0$ such that for all $n \geq N_0$, 
\begin{equation} \label{eq5} 
   n \geq 16^d c^{d+1} \ln^d n .     
\end{equation} 
We claim that $N_0 = e^{\alpha (d+1)!}$, where $\alpha = 16^d c^{d+1}$,
has this property. To prove this, let $m \geq \alpha (d+1)!$. Then 
\[ e^m = \sum_{k=0}^{\infty} \frac{m^k}{k!} \geq \frac{m^{d+1}}{(d+1)!}
      \geq \alpha m^d 
\]
and, thus, if $n \geq N_0$, 
\[ n = e^{\ln n} \geq \alpha \ln^d n . 
\] 
Define $A$ to be the maximum of $2 c^2$ and 
\[ \max \left\{ \frac{T(k)}{k \ln k} : 2 \leq k < N_0 \right\} . 
\] 
Observe that $A$ only depends on $d$. 

We will prove that for all integers $n$ with $2 \leq n \leq N$,
\begin{equation}   \label{eq6} 
   T(n) \leq A n \ln n . 
\end{equation} 
The proof is by induction on $n$. If $2 \leq n < N_0$, then (\ref{eq6})
follows from the definition of $A$. 

Let $n \geq N_0$, and assume that (\ref{eq6}) holds for all values less 
than $n$. Let $n'$ and $n''$ be two integers that satisfy 
(\ref{eq1})--(\ref{eq3}). By the induction hypothesis, we have 
\[ T(n') \leq A n' \ln n' \leq A n' \ln ( (1-1/c)n ) 
\]
and  
\[ T(n'') \leq A n'' \ln n'' \leq A n'' \ln ( (1-1/c)n ) ,
\]
implying that 
\[ T(n') + T(n'') \leq A (n' + n'') \ln ( (1-1/c)n )  
  \leq A ( n + n/t ) \ln ( (1-1/c)n ) . 
\]
From (\ref{eq4}), we get
\begin{eqnarray*} 
  T(n') + T(n'') & \leq &  
      A \left( n + 8 c^{1/d} n^{1-1/d} \right)  \ln ( (1-1/c)n ) \\ 
  & = & A n \ln n + A n \ln(1-1/c) + 
        8A c^{1/d} n^{1-1/d} \ln ((1-1/c)n) \\ 
  & \leq & 
   A n \ln n + A n \ln(1-1/c) + 8A c^{1/d} n^{1-1/d} \ln n \\ 
  &\leq  & A n \ln n - A n/c + 8A c^{1/d} n^{1-1/d} \ln n  ,
\end{eqnarray*} 
where in the last step we used the inequality 
$\ln(1 - x) \leq -x$, which is valid for all 
real numbers $x$ with $x < 1$.
By the definition of $A$, we have $A \geq 2 c^2$, implying that 
\[ A/c - c \geq A/(2c) . 
\] 
Thus, 
\[ cn + T(n') + T(n'') \leq 
   A n \ln n - A n /(2c) + 8A c^{1/d} n^{1-1/d} \ln n . 
\] 
By (\ref{eq5}), we have 
\[ n^{1/d} \geq 16 c^{1+1/d} \ln n 
\]
and, therefore, 
\[ 8A c^{1/d} n^{1-1/d} \ln n \leq An/(2c) . 
\] 
We conclude that 
\[ cn + T(n') + T(n'') \leq A n \ln n . 
\] 

Since $n'$ and $n''$ were arbitrary integers satisfying 
(\ref{eq1})--(\ref{eq3}), we have shown that (\ref{eq6}) 
holds for the current value of $n$. Thus, (\ref{eq6}) holds 
for all integers $n$ with $2 \leq n \leq N$.

\section{Concluding remarks} 
We have presented a very simple randomized algorithm for computing the 
closest-pair distance in metric spaces of small doubling dimension. 
The algorithm only uses the following operations: 
\begin{enumerate} 
\item For any given point $p$, count or determine all points that are 
      within a given distance from $p$, or within a given range of 
      distances from $p$. This operation can obviously be done in 
      linear time, by simply scanning the sequence of distances between 
      $p$ and all points.  
\item For a given sequence of $n$ real numbers, find the $k$-th smallest 
      element in this sequence. This operation can be done in expected 
      linear time, again by a simple randomized algorithm; see 
      Cormen \emph{et al.}~\cite[Chapter~9]{clrs-ia-09} and 
      Kleinberg and Tardos~\cite[Section~13.5]{kt-ad-06}. 
\end{enumerate} 

\section*{Acknowledgements}
This research was carried out at the \emph{Eighth Annual Workshop on 
Geometry and Graphs}, held at the Bellairs Research Institute in 
Barbados, January 31 -- February 7, 2020. The authors are grateful to 
the organizers and to the participants of this workshop. 

\bibliographystyle{plain}
\bibliography{ClosestPairDoubling}

\begin{thebibliography}{1}

\bibitem{ahp-ncsspd-12}
M.~A. Abam and S.~Har-Peled.
\newblock New constructions of {SSPDs} and their applications.
\newblock {\em Computational Geometry: Theory and Applications}, 45:200--214,
  2012.

\bibitem{a-pldr-83}
P.~Assouad.
\newblock Plongements lipschitziens dans $\mathbb{R}^{N}$.
\newblock {\em Bulletin de la Soci{\'{e}}t{\'{e}} Math{\'{e}}matique de
  France}, 111:429--448, 1983.

\bibitem{b-phd-76}
J.~L. Bentley.
\newblock {\em Divide and Conquer Algorithms for Closest Point Problems in
  Multidimensional Space}.
\newblock Ph.{D}. thesis, Department of Computer Science, University of North
  Carolina, Chapel Hill, N.C., 1976.

\bibitem{b-cacm-80}
J.~L. Bentley.
\newblock Multidimensional divide-and-conquer.
\newblock {\em Communications of the ACM}, 23:214--229, 1980.

\bibitem{bs-dcms-76}
J.~L. Bentley and M.~I. Shamos.
\newblock Divide-and-conquer in multidimensional space.
\newblock In {\em Proceedings of the 8th ACM Symposium on the Theory of
  Computing}, pages 220--230, 1976.

\bibitem{clrs-ia-09}
T.~H. Cormen, C.~E. Leiserson, R.~L. Rivest, and C.~Stein.
\newblock {\em Introduction to Algorithms}.
\newblock MIT Press, 3rd edition, 2009.

\bibitem{hm-fcnld-05}
S.~Har-Peled and M.~Mendel.
\newblock Fast construction of nets in low-dimensional metrics and their
  applications.
\newblock {\em SIAM Journal on Computing}, 35:1148--1184, 2006.

\bibitem{h-lams-01}
J.~Heinonen.
\newblock {\em Lectures on Analysis on Metric Spaces}.
\newblock Springer-Verlag, 2001.

\bibitem{kt-ad-06}
J.~Kleinberg and E.~Tardos.
\newblock {\em Algorithm Design}.
\newblock Addison-Wesley, 2006.

\end{thebibliography}

\end{document}